\documentclass[11pt,letterpaper]{article}

\usepackage[margin=0.90in]{geometry}
\usepackage[numbers,sort&compress]{natbib}
\usepackage{caption}
\usepackage{authblk}
\usepackage{placeins}
\usepackage[T1]{fontenc}
\usepackage{lmodern}
\usepackage{amsmath,amssymb,mathtools,amsthm,bm}
\usepackage{microtype}
\usepackage{booktabs}
\usepackage{graphicx}
\usepackage{array}
\usepackage{enumitem}
\usepackage{xcolor}
\usepackage[hidelinks]{hyperref}
\hypersetup{pdftitle={Certified Bounds on the Measurement Cost of Molecular Hamiltonians},pdfauthor={Federico Zahariev and Vassiliki-Alexandra Glezakou}}
\usepackage{tikz}
\usetikzlibrary{arrows.meta,positioning,fit,calc,shapes.geometric}

\newtheorem{theorem}{Theorem}
\newtheorem{proposition}{Proposition}
\newtheorem{lemma}{Lemma}

\theoremstyle{definition}

\newcommand{\D}{\mathcal D}

\newcommand{\E}{\mathbb E}

\newcommand{\Tr}{\operatorname{Tr}}
\newcommand{\Var}{\operatorname{Var}}
\newcommand{\Cov}{\operatorname{Cov}}
\newcommand{\rank}{\operatorname{rank}}
\newcommand{\range}{\operatorname{range}}

\newcommand{\norm}[1]{\left\lVert #1\right\rVert}

\newcommand{\eps}{\varepsilon}

\usepackage{needspace}
\usepackage{pdfpages}
\usepackage{textcomp}
\begin{document}
\title{Certified Bounds on the Measurement Cost of Molecular Hamiltonians}
\author[1,2]{Federico Zahariev}
\author[1]{Vassiliki-Alexandra Glezakou}
\affil[1]{Chemical Sciences Division, Oak Ridge National Laboratory, Oak Ridge, Tennessee 37831, USA}
\affil[2]{Department of Chemistry and Ames National Laboratory, Iowa State University, Ames, Iowa 50011, USA}
\date{October 7, 2026}
\maketitle
\begin{center}\small Correspondence: \href{mailto:fzahari@iastate.edu}{fzahari@iastate.edu}. \href{mailto:glezakouva@ornl.gov}{glezakouva@ornl.gov}\end{center}
\begin{abstract}Quantum estimates of molecular energy require repeated state preparation and measurement. We bound the minimum repetition count for given measurements and processing rules from above with a practical schedule and from below with independently verified inequalities. For an H$_4$ chain at H--H distances of 0.90 and 1.80~\AA\ with 2 per cent uniform noise in the ground state, entangling Clifford measurements, which map Pauli products onto one another, cut leading-order repetitions by at least 65.9 and 61.0 per cent, respectively, relative to any schedule using only single-qubit Pauli measurements and Hamiltonian terms. With 69 unaltered measurement settings, each a basis change followed by qubit readout, processing further observables saves at least 12.04 and 6.70 per cent. For four lanthanide and actinide Hamiltonians with depolarized Hartree-Fock covariances, reducing the sum of absolute Hamiltonian coefficients does not proportionally reduce sampling. The bounds cover linear unbiased estimates from independent batches at small error.
\end{abstract}
\section{Introduction}\label{sec:intro}
How the Hamiltonian is mapped to the measurements at hand dictates the statistical price of determining a molecular energy. Every \emph{shot}, a cycle of state preparation and measurement, is a source of random outcomes. One can average more shots to bring down the statistical error, but the number needed is also a function of how shots are allocated across settings and the correlations between observables measured together.

There is room for improvement in any proposed procedure, and a lower bound on cost makes this quantifiable. The literature is replete with methods that optimize outcome processing by grouping compatible Pauli terms and sharing coefficients where groups overlap~\cite{Verteletskyi2020MCC,Yen2023Deterministic,Wu2023OverlappedGrouping,Huang2020ClassicalShadows,Fischer2024DualFrames,Choi2023Fluid,Patel2025Review}. A Pauli term, to be precise, is a tensor product of identity and $X$, $Y$ or $Z$ matrices times a Hamiltonian coefficient. Yen, Ganeshram and Izmaylov have shown how to lower the variance of overlapping-group estimates via covariances and coefficient splitting~\cite{Yen2023Deterministic}, while Choi, Yen and Izmaylov found that observables not in the Hamiltonian can do the same if their coefficients cancel out in different groups~\cite{Choi2022Ghost}. Our approach is to optimize these options in concert and to bound the cost of every estimate the model allows.

Any measurement model has to spell out the physical procedure as well as the information to be had from it. We refer to the procedure, typically a basis change and qubit readout, as a \emph{setting}. To each outcome a \emph{score} is assigned, a real number whose mean is the expectation of some observable, for instance the eigenvalue of a Pauli string obtained by multiplying the outcomes of chosen qubits. With its scores, their covariance matrix and per-shot cost, a setting becomes what we call a \emph{context}, and a set of candidate contexts a \emph{landscape}. A \emph{schedule} is then the distribution of Hamiltonian coefficients among the scores and the assignment of an independent batch of shots to each setting in use.

We turn to convex experimental design for both an attainable cost and a floor for this model. In classical $c$-optimal design the goal is the variance of a particular linear combination of correlated responses~\cite{Elfving1952,DetteHollandLetz2009,Sagnol2011}. In our case those responses are the measurement scores whose linear combination reconstructs the Hamiltonian. An optimized schedule is our upper bound. For the lower bound we employ a \emph{dual witness}, a vector that meets certain inequalities and bounds all admissible schedules from below. We verify the witness and the schedule on our own, apart from the optimizer, and can extend the method to covariances gleaned from preliminary data. The witness will tell us if reoptimization is warranted for new settings or scores. It is not something Pauli-grouping implementations tend to export routinely for every estimator in a given collection, so far as we are aware. We are not putting forward conic duality as a novelty. What we offer is its application to a quantum-measurement environment where collections may fall short of being informationally complete, with explicit scores, stratified allocation, joint coefficient sharing and a sparse handling of large Hamiltonians. Section~\ref{sec:relations} puts the construction in relation to current practice.

The paper is laid out as follows: Theorem~\ref{thm:main} brings together shot allocation and coefficient sharing in a second-order cone program. Theorem~\ref{thm:dual} matches a schedule with a checked bracket $L\leq\Phi\leq U$ and a dual witness. Theorem~\ref{thm:robust} widens the bracket to include covariances bounded by preliminary measurements. Propositions~\ref{prop:ghost-extension} and~\ref{prop:colgen} serve to test whether a shared score holds the witness or if additional settings are merited by their dual prices. We present the molecular evidence in Secs.~\ref{sec:complete-family} to \ref{sec:felement}, with Sec.~\ref{sec:strategies} making the distinction between a nearly optimal schedule, inefficient use of useful settings and the limitations of the collection.

Our analysis is of the sampling route to expectation values, be it variational, analog or repeated-preparation. It does not reduce the phase-estimation costs of qubitization or Trotterized walk-operator algorithms, which are without a Pauli-measurement schedule. Then again, the simulation and subsequent measurement problem can be affected by orbital changes~\cite{Koridon2021}, symmetry shifts~\cite{Loaiza2023BLISS,Patel2025BLISS}, tapering~\cite{Bravyi2017Tapering} or truncation, and we consider these in Sec.~\ref{sec:felement}. In~\cite{ZaharievOrtiz2026Factorization} the generic factor $\alpha_O^2/\eps^2$ gives way to an optimized sampling factor $\Phi^2/\eps^2$ ($\alpha_O$ being the coefficient one-norm of the observable), all without any change to the simulation factor.

The two hydrogen-chain models demonstrate a sampling advantage over the complete local Pauli family. Each is an H$_4$ system in minimal STO-3G, meaning three contracted Gaussian functions stand in for every Slater-type atomic orbital across the four hydrogen atoms. The eight-qubit encodings for these models open up 6561 possible assignments of $X$, $Y$ or $Z$ readout to individual qubits. If one confines the scores to Hamiltonian terms, the inclusion of entangling Clifford settings will cut the leading shot cost by no less than 65.9\% and 61.0\%, respectively. This is because Clifford conjugation can map Pauli strings onto one another so that a commuting group may be read at once. A ten-qubit LiH run offers yet another point of comparison with the complete family. By putting in scores at 69 of the fixed H$_4$ settings, one can save at least 12.04\% and 6.70\%.

Then there are the f-element models, which make the distinction between a reduction in the coefficient one-norm and in sampling. Take the CeO, NdO, UO$_2^{+}$ and UO$_2^{2+}$ operators: before tapering they require 30 to 36 qubits, and 27 to 32 after. Tapering disposes of qubits fixed by the symmetries whose shifts keep the operator intact in a chosen sector but alter it in others. When run under a depolarized Hartree-Fock covariance model, the full shift brings down the coefficient one-norm (the sum of absolute Pauli coefficients excluding the identity term) by 49--65\% and the leading sampling cost some 24--40\%. Tapering leaves the coefficient one-norm unchanged for these operators, but it does reduce the leading cost by around 8--31\%.

The guarantees apply to linear unbiased estimates of a single observable from independent measurement batches, so their averages will track the target expectation. The associated costs are a function of the positive per-shot costs, covariances and the scores in question. The \emph{leading} cost is what is proportional to the inverse squared error as the latter goes to zero. Complete-family work relies on discrete Pauli axes and Hamiltonian-term scores, and the records here note the floating-point tolerances and any numerical feasibility repairs.

\section{Results}\label{sec:results}
\subsection{Upper and lower bounds on sampling cost}\label{sec:exact}\label{sec:dual}\label{sec:contexts}
An optimization yields an executable schedule along with a bound on how much further it can be improved. One need only take out the identity part of $H$, since it has no variance, and put the rest of the coefficients in $h$. For $M$ settings, $x_\ell$ is the assignment of coefficients to the scores, $\Gamma_\ell$ their covariance matrix and $A_\ell$ puts the contributions in the common coordinate ordering of $h$. The fragment $F_\ell$ that results is required to be part of an exact reconstruction:
\begin{equation}
H=\sum_{\ell=1}^{M}F_\ell,
\qquad h=\sum_{\ell=1}^{M}A_\ell x_\ell.
\label{eq:fragment_constraint}
\end{equation}

Variances from independent measurement batches are additive. Put $N_\ell$ shots to setting $\ell$ and the energy-estimator variance is
\begin{equation}
\Var(\widehat H)=\sum_{\ell=1}^{M}\frac{x_\ell^T\Gamma_\ell x_\ell}{N_\ell}.
\label{eq:estimator_variance}
\end{equation}
The execution cost in total is $\sum_\ell c_\ell N_\ell$ with $c_\ell>0$ per shot. Set $c_\ell$ to unity and the cost is the shot count.

\begin{theorem}[Optimal allocation and sampling cost]\label{thm:main}
Given a finite set of settings whose scores reconstruct the target and assuming positive per-shot costs $c_\ell$ and positive-semidefinite covariance matrices $\Gamma_\ell\succeq0$ (ensuring nonnegative variance in every fragment), the minimum leading cost for a linear unbiased estimate from independent batches with mean-square error not exceeding $\eps^2$ is
\begin{equation}
C_{\rm opt}(\eps)=\frac{\Phi^2}{\eps^2},\label{eq:cost_from_phi}
\end{equation}
where the measurement constant is
\begin{equation}
\begin{aligned}
\Phi=\min_{x_1,\ldots,x_M}\quad&\sum_\ell\sqrt{c_\ell}\sqrt{x_\ell^T\Gamma_\ell x_\ell}\\
\text{subject to}\quad&\sum_\ell A_\ell x_\ell=h.
\end{aligned}
\label{eq:primal}
\end{equation}
and for an optimal split of the coefficients the continuous shot allocation is
\begin{equation}
N_\ell=\frac{\Phi}{\eps^2}\sqrt{\frac{x_\ell^T\Gamma_\ell x_\ell}{c_\ell}},
\label{eq:shot_allocation}
\end{equation}
Rounding the optimized figures up to integers affects only lower-order terms when $\eps$ is small.
\end{theorem}

This is a matter of convex optimization, with coefficient sharing and shot allocation handled together. Reconstruction is a linear constraint and each objective term a standard deviation. An equivalent second-order cone program (SOCP) is obtained by factoring $\Gamma_\ell=R_\ell^TR_\ell$ and introducing an auxiliary variable to bound the norm:
\begin{equation}
\begin{aligned}
\min_{x,t}\quad&\sum_\ell\sqrt{c_\ell}\,t_\ell\\
\text{subject to}\quad&\sum_\ell A_\ell x_\ell=h,\qquad \|R_\ell x_\ell\|_2\leq t_\ell.
\end{aligned}
\label{eq:socp}
\end{equation}
A convex solver will do the job directly rather than having to alternate updates. Section~\ref{app:proofs} covers the allocation in the small-error limit for deterministic scores.

Any schedule is bounded from below by a proper weighting $y$ of the reconstruction equations, where admissible weights are a matter of the inverse covariance. To handle singular covariances we define
\begin{equation}
\|z\|_{\Gamma,*}=\begin{cases}\sqrt{z^T\Gamma^\dagger z},&z\in\operatorname{range}\Gamma,\\+\infty,&\text{otherwise}.\end{cases}
\label{eq:dual_seminorm}
\end{equation}
Here the pseudoinverse $\Gamma^\dagger$ inverts the nonzero eigenvalues and leaves the zero ones be. There is also a range condition to rule out any $z$ component in a zero-variance direction. Without it one might end up with a vector and an untenable lower bound.

\begin{theorem}[Dual witness and cost bracket]\label{thm:dual}
The dual problem, under the assumptions of Theorem~\ref{thm:main}, is
\begin{equation}
\sup_y y^Th\quad\text{subject to}\quad
\|A_\ell^Ty\|_{\Gamma_\ell,*}\leq\sqrt{c_\ell}\quad\text{for every }\ell.
\label{eq:dual}
\end{equation}
and any $y$ that is feasible will serve as a lower bound. Once deterministic directions are discarded, the relative-interior condition on the span of the target and score means the primal and dual optima are not only equal but attained. This is relative to the subspace the model permits. So one has a feasible schedule with value $U$ and a witness of value $L=y^Th>0$ which gives
\begin{equation}L\leq\Phi\leq U.\label{eq:bracket}\end{equation}
In other words, the leading cost of the schedule is at most $(U/L)^2$ times what the model's minimum would be.
\end{theorem}

One can verify the witness for each setting in isolation. Even if the solver has not converged, putting the dual inequalities together with Eq.~\eqref{eq:fragment_constraint} shows $y^Th\leq\Phi$ for any reconstructing schedule. For the Lagrangian derivation see Section~\ref{app:dual}. What remains between the bounds we have checked is the improvement the model has yet to rule out. Should an upper bound fall under another model's lower bound, an advantage over the whole of that comparison model is established. With the same cost convention and precision, suppose $0<L_A\leq\Phi_A\leq U_A$ and $0<L_B\leq\Phi_B\leq U_B$. The resulting saving in leading cost is

\begin{equation}
1-\left(\frac{U_B}{L_A}\right)^2
\leq1-\left(\frac{\Phi_B}{\Phi_A}\right)^2
\leq1-\left(\frac{L_B}{U_A}\right)^2.
\label{eq:paired-saving}
\end{equation}
A strictly positive minimum sampling saving is assured when $U_B<L_A$. Take a single schedule where the numerical gap is $g=U/L-1$. Then its leading cost exceeds the model's minimum by no more than $(1+g)^2-1=2g+g^2$.

Before issuing a certificate, we check numerical feasibility and repair it where we can. The verifier will correct the reconstruction, look at covariance ranges, rescale the admissible dual vectors and work out the objectives from the inputs on file. These are floating-point checks rather than directed-rounding interval calculations, so the outward rounding of the values put on display does not account for every arithmetic error. Residuals and tolerances are set out in Section~\ref{app:repair}.
\subsection{Comparison with every local Pauli setting}\label{sec:complete-family}
For the H$_4$ geometries, entangling readout comes in at a cost under the lower bound of the full local Pauli family in our model. We consider linear chains at uniform 0.90 and 1.80 Å spacings with STO-3G. The four electrons in four spatial orbitals amount to eight Jordan-Wigner qubits and 184 nonidentity terms in the Hamiltonian. With shot costs of one, covariances are taken as $0.98|\Psi_0\rangle\langle\Psi_0|+0.02I/256$ ($\Psi_0$ being the exact ground state for the four electrons).

This bound from the local family is comprehensive, covering the $3^8=6561$ ways of assigning $X$, $Y$ or $Z$ axes to the eight qubits. Those terms that are readable in a given setting and commute on each qubit are the qubit-wise commuting (QWC) ones. We include only their Hamiltonian-term scores. One could define a larger model by permitting continuously rotated local axes or more observable scores.

There is no need to keep redundant settings to find the optimum. Where two settings have the same list of readable target terms, the score covariances are the same. Lemma~\ref{lem:dominance} shows that at equal cost a fragment on a smaller list can be moved to a containing list at no extra objective. In this way the 6561 settings of the complete family are whittled down to 2135 distinct lists, of which 565 are maximal with 3904 coefficient copies.

By uniting local and entangling settings, one is assured of leading shot savings of at least 65.9\% and 61.0\% over the complete local family. The sorted insertion method yields ten fully commuting (FC) groups at 0.90 Å and nine at 1.80 Å. Because their full Pauli strings commute regardless of the individual-qubit factors, a Clifford circuit can be used for each readout. Even with FC-only optimization, savings of 41.8\% and 41.4\% are guaranteed, though Table~\ref{tab:complete-family} shows the larger reductions come from sharing coefficients between FC and local settings.

\begin{table}[htbp]\centering\small
\caption{Bounds on the optimized measurement constant $\Phi$, in hartree, are compared for two H$_4$/STO-3G chains. All models use unit shot costs, Hamiltonian-term scores, and state $0.98|\Psi_0\rangle\langle\Psi_0|+0.02I/256$, with $\Psi_0$ the exact four-electron ground state. The upper bound $U_{69}$ uses 69 selected local settings, and $L_{\rm prod}$ bounds the full family of 6561 local Pauli settings from below. The upper bounds $U_{\rm FC}$ and $U_{\rm union}$ use fully commuting (FC) settings and their union with the complete local family, respectively. Relative gaps are below $6\times10^{-6}\%$. The final column gives $1-(U_{\rm union}/L_{\rm prod})^2$, rounded down. Bounds are rounded outward after the floating-point checks of Sec.~\ref{app:repair}.}
\label{tab:complete-family}
\begin{tabular}{rrrrrl}\toprule
Spacing (\AA)&$U_{69}$&$L_{\rm prod}$&$U_{\rm FC}$&$U_{\rm union}$&Minimum shot saving\\\midrule
0.90&1.626070&1.544314&1.177186&0.901217&$\geq65.9\%$\\
1.80&1.558184&1.384513&1.059291&0.864002&$\geq61.0\%$\\\bottomrule
\end{tabular}
\end{table}

\begin{figure}[htbp]\centering
\includegraphics[width=0.92\linewidth]{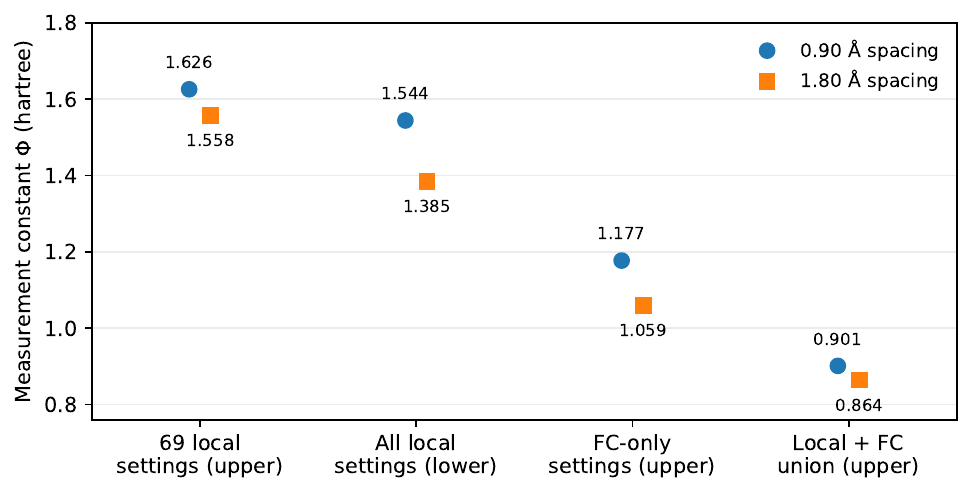}
\caption{For the H$_4$/STO-3G models in Table~\ref{tab:complete-family}, entangling readout reaches measurement constants below the lower bound for every local Pauli setting. Markers give upper values for 69 selected local settings, fully commuting (FC) settings, and the combined family. The complete local family provides the lower bound over 6561 settings. Covariances use exact ground states mixed with 2\% uniform noise, with equal shot costs and Hamiltonian-term scores only. Values are in hartree, and squared ratios give leading shot ratios.}
\label{fig:family}
\end{figure}

The lower bounds put a cap on what can be gained by swapping out the chosen local settings. The optimized 69-setting values sit 5.3\% and 12.5\% above the complete-family floor, meaning any local Pauli collection under the same score rule can save no more than 9.8\% and 21.1\% relative to the optimized 69-setting schedules. To get better numbers, additional scores or settings beyond that family are called for. In fact, a schedule that reads commuting groups jointly will outperform the whole local family when it comes to Hamiltonian-term scores. One might view this as the molecular equivalent of the structural separation seen in Ref.~\cite{Zahariev2026Clifford} for orbital-rotation and Clifford measurements: there, a Bell-diagonal target demands three orbital settings but only one Clifford setting, albeit under the different physical constraint of separately conserving the two spin populations.

In the case of ten-qubit LiH, the inclusion of FC settings with the complete local family secures at least 34.9\% in leading shot savings. While $3^{10}=59049$ local settings could be used to cover the 275 target terms, they can be whittled down to 15221 distinct lists (2509 of them maximal, with 23363 coefficient copies). Applying 2\% depolarization and the $10^{-8}I$ covariance regularization from Sec.~\ref{sec:molecular} puts the downward-rounded lower bound at $0.7063$ hartree. The 80 local settings in Supplementary Table S1 have an optimized value of $0.7823$, some 10.8\% higher. The 20 FC settings on their own produce an upper value of $0.5952$ ($0.5951$ before rounding in the table), and their union with the full local family comes in at $0.5696$, an improvement over the $0.5879$ for 80 local settings with the same FC component. As for computational time, complete-family optimization runs in about 16 seconds on a single central processing unit (CPU) core for LiH and a second or so for H$_4$.
\subsection{Additional information from unchanged settings}
\label{sec:ghost-pricing}
There is still variance to be reduced by an observable not in the Hamiltonian if it is read across several settings. Choi, Yen and Izmaylov have termed such observables ``ghost'' Pauli products~\cite{Choi2022Ghost}. Here we offer a test to see if one of these scores can be admitted without losing the current lower-bound witness. The coefficient copies are picked to cancel in the reconstructed target, yet their covariances with Hamiltonian scores serve to dampen fragment fluctuations.

Introducing a shared score brings in a common observable coordinate with a zero total target coefficient. With the present scores and shot costs held constant, the covariance is expanded to
\begin{equation}
\widetilde\Gamma_\ell=
\begin{pmatrix}\Gamma_\ell&g_\ell\\g_\ell^T&s_\ell\end{pmatrix}\succeq0,
\label{eq:ghost-covariance}
\end{equation}
in which $g_\ell$ holds the covariances with the new score and $s_\ell$ its variance. For an old feasible witness $y$, set $z_\ell=A_\ell^Ty$.

\begin{proposition}[Adding a score to an existing witness]
\label{prop:ghost-extension}
Let
\begin{align}
q_\ell&=s_\ell-g_\ell^T\Gamma_\ell^\dagger g_\ell, &
m_\ell&=g_\ell^T\Gamma_\ell^\dagger z_\ell,\\
r_\ell&=\sqrt{q_\ell\bigl(c_\ell-z_\ell^T\Gamma_\ell^\dagger z_\ell\bigr)}.
\end{align}
The witness can be extended to the added score if and only if
\begin{equation}
\max_\ell(m_\ell-r_\ell)\leq\min_\ell(m_\ell+r_\ell),
\label{eq:ghost-interval-test}
\end{equation}
where the extrema run over all applicable settings. Any number in the interval will do for the new dual coordinate. If the original bounds are coincident, a feasible extension is proof that the optimum has not moved.
\end{proposition}

Since the dual coordinate must hold for every affected setting at once, Eq.~\eqref{eq:ghost-interval-test} is employed to check whether the highest lower endpoint is greater than the lowest upper endpoint. We derive the intervals for singular and nonsingular blocks in Sec.~\ref{method:score-proof} and have run implementation tests on 1000 random block-covariance examples plus a singular instance. Should the interval test fail, the score has invalidated the witness and reoptimization is in order. This constitutes strict improvement only if none of the old optimal witnesses can be extended (dual attainment assumed), though another may still be feasible. Of course, with several new scores imposing coupled quadratic constraints, an individual test is not enough.

\subsubsection{An exact example at unchanged settings}
\label{sec:ghost-example}
Additional processing alone can give a 36\% leading shot reduction in a two-qubit scenario. Take
\begin{equation}
\rho=\tfrac14(I+\tfrac35 XZ-\tfrac35 ZZ),\qquad H=XI+ZI.
\end{equation}
The state's positivity is confirmed by the eigenvalues $(1\pm3\sqrt2/5)/4$. At first, $XZ$ and $ZZ$ settings yield only the $XI$ and $ZI$ Hamiltonian scores, each with unit variance and zero mean, for $\Phi_{\rm old}=2$ and a corresponding witness $y=(1,1)$.

Put $IZ$ into both settings and the fragment variances drop. Its variance is one, with $3/5$ and $-3/5$ covariances against the $XI$ and $ZI$ scores. An extension of the old witness would demand the new coordinate be equal to both, so the interval test is not passed. Reoptimization leads to
\begin{equation}
F_1=XI-\tfrac35 IZ,\qquad F_2=ZI+\tfrac35 IZ.
\end{equation}
with fragments that add up to $H$ and have a variance of $16/25$. The dual witness $(4/5,4/5,0)$ is in accord with the new $\Phi_{\rm new}=8/5$, so
\[
1-(\Phi_{\rm new}/\Phi_{\rm old})^2=9/25=36\%.
\]
This is verified by independent numerical work and an exact rational check.

\subsubsection{Additional scores for hydrogen chains}\label{sec:molecular-ghost}
We have found twelve extra scores that bring the sampling cost for H$_4$ down without disturbing the 69 local settings. Nonidentity products of compatible Hamiltonian terms were generated and those readable at two or more settings kept. The twelve with the most separation in the witness intervals were the ones admitted. Each is put to every compatible setting, and reconstruction requires its total coefficient to vanish. Table~\ref{tab:ghost-h4} has the reoptimized figures.

\begin{table}[htbp]\centering
\caption{Twelve shared scores reduce the H$_4$/STO-3G measurement constant at 69 fixed local settings. At each geometry, both models have unit shot costs and covariance state $0.98|\Psi_0\rangle\langle\Psi_0|+0.02I/256$, where $\Psi_0$ is the exact four-electron ground state. Upper values $U_{\rm old}$ and $U_{\rm new}$ are in hartree. The minimum leading shot saving is $1-(U_{\rm new}/L_{\rm old})^2$, rounded downward. Relative gaps $U/L-1$ are below $1.7\times10^{-7}$ under the numerical convention in Sec.~\ref{app:repair}.}\label{tab:ghost-h4}
\begin{tabular}{rrrr}\toprule
Spacing (\AA)&$U_{\rm old}$&$U_{\rm new}$&Minimum shot saving\\\midrule
0.90&1.626070&1.525002&$\geq12.04\%$\\
1.80&1.558184&1.505049&$\geq6.70\%$\\\bottomrule
\end{tabular}\end{table}

In the end, whether enriching the scores is better than changing the settings is a matter of the H$_4$ geometry. Enrichment at 0.90~\AA\ yields $1.525002$, a figure that falls short of the $1.544314$ complete local-family bound for Hamiltonian-term scores. The picture is different at 1.80~\AA\ where the enriched value of $1.505049$ is in excess of the $1.384513$ bound. One can see from this that altering local settings is more productive at the greater spacing than the twelve-score addition we have put to the test.

With the 565-setting model, the introduction of 511 scores brings the value down to $1.2722$ at 0.90 Å. While this represents an upper bound for the complete local family (enrichment has a way of putting previously redundant settings back into use), it cannot be taken as the lower bound for a comparison with Clifford readout under unrestricted processing. That is because no single score left at termination fails the interval test, yet a certificate for the expanded family would demand a simultaneous witness extension across the board. Supplementary Note S5 goes into the restrictions on this. In any case, the union value of $0.9012$ in Table~\ref{tab:complete-family} is some 29\% lower in $\Phi$. Even the twelve-score enrichment of the 69-setting model, which reached $1.525$, captures only a small part of what is possible with score enrichment over the full family.
\subsection{Uncertain covariances and molecular benchmarks}\label{sec:benchmarks}
It is possible to put a limit on the optimal sampling cost from preliminary data even without knowing the precise covariance of the prepared state. To do so we build upper and lower matrices that enclose the setting covariances at the given confidence level. The upper ones will bound the cost of an executable schedule and the lower ones the dual witness, producing the interval of Theorem~\ref{thm:robust} (six examples are in Supplementary Table S3). If small covariance eigenvalues are not well resolved, the lower bound on $\Phi$ may be diminished or even nullified. We have used the proved radius for all our reported intervals, though an empirical calibration on the twelve LiH settings gives intervals 10.4 to 10.6 times tighter at 95\% empirical coverage (Sec.~\ref{sec:robust}).

\subsubsection{Improving a schedule within its existing settings}\label{sec:strategies}
The joint optimization of coefficients and shots does away with a good deal of the superfluous cost in the schedules we have examined. For a given schedule $S$ we hold the actual settings and scores and work out the optimized upper value $U_{\rm own}$. The expression $(\Phi_S/U_{\rm own})^2-1$ then gives a conservative lower bound on the excess leading cost. In the case of H$_2$O, the sorted-insertion QWC and FC schedules and the weight-greedy local one run to about 2.1--2.7 times their optimized values. The largest-degree-first FC grouping is far more expensive, with excesses of 666\% for H$_2$O and 767\% for LiH. A comparison with the QWC+FC reference can be found in Supplementary Tables S1 and S2.

Supplementary Note S7 puts these reference values up against the variances for H$_2$, LiH, BeH$_2$ and H$_2$O published by Yen, Ganeshram and Izmaylov~\cite{Yen2023Deterministic}. Their best approach, iterative coefficient splitting, still demands 1.1 to 1.6 times the optimized leading shot cost with FC settings for the latter three molecules. Sorted insertion is 2.6 to 6.6 times the cost.

\subsubsection{Methane benchmark}\label{sec:ch4}
The CH$_4$ benchmark sees a leading shot cut of around 39\% when FC settings are brought in to the local measurements. We are dealing with an STO-3G Hamiltonian of 18 qubits and 6891 nonidentity terms. Using a Hartree-Fock reference with 2\% depolarization for the covariances, and allowing unit shot costs and no more than 80 target-term scores per setting, we obtain upper values of $10.795$ hartree for QWC and $8.455$ for QWC+FC (relative gaps of 1.3 and $1.7\times10^{-4}\%$). The squared ratio of these is the shot comparison. This operator is also the parity-ceiling example in the companion paper~\cite{Zahariev2026Clifford}, and the comparison here is a measure of the price for this particular FC addition. We have not put together a production orbital-rotation collection, so the result does not compare FC settings with orbital-rotation measurements. Optimization of the original sorted-insertion FC and QWC schedules within their own settings trims the costs by 3.05 and 1.87.

\subsubsection{Selecting fewer settings and testing additions}\label{sec:dictdesign}\label{sec:colgendemo}
RANGE, the Robust Adaptive Nature-inspired Global Explorer~\cite{RANGE}, is employed to pick candidate collections for LiH, BeH$_2$ and H$_2$O, after which a conic solve handles the coefficient sharing and shot allocation. As shown in Supplementary Table S4, reducing the size of the collections by factors of 4.3, 6.1 and 5.8 leaves the measurement constants within 0.2\%, 1.0\% and 2.1\% of what they were originally. The cost is bounded by the certificate within its chosen collection, yet one is left with a nonconvex search when it comes to picking among collections. Because of coefficient sharing, settings are interdependent. An independent score for each will not in general reveal the optimal subset. For example, for LiH, 23 settings from the search give an increase of 0.2\%. A greedy approach with 40 settings produces 1.4\% despite holding on to 74\% more settings, and the 19-setting version in Supplementary Table S4 1.5\%.

In a RANGE run that looked at collection size and cost as one and held onto nondominated candidates, the LiH excess costs came in at 1.2\% for 22 settings, 0.6\% for 25, and between 10 and 14\% for 19 to 21. These figures pertain to what was found in that particular run. The superior 19-setting result from the greedy method indicates the threshold is an artefact of the run and not representative of the finest combinations of size and cost.

To see if the convex portion of the problem ought to be put to the search, a control was run. A direct evolutionary search on LiH coefficient splittings under the reconstruction constraint gave $\Phi=25.10$ (versus $0.7823$ for the conic value), making it some 32 times worse in $\Phi$ and twice as time consuming, with no dual witness to show for it. While this single control does not rank search algorithms, it supports leaving shot allocation and coefficient sharing to the convex solve and let RANGE handle the nonconvex selection.

We also extended the LiH certificate for 80 local settings to a fixed pool of 20 FC candidates. Here the active problem is reoptimized after testing each candidate against the witness and adding the two with the most violations. Over ten rounds the active value drops from $0.7823$ to $0.5879$ hartree. By round nine the largest pool price has come down from $13.6$ to $1.15$ with only two candidates left. Once they are in, the pool is empty and the condition $\nu_y^\star\leq1$ is met. Proposition~\ref{prop:colgen} then certifies $0.5879$ over the whole pool with $\eta=0$, in line with a full pool optimization. Raw solver gaps do not exceed $3\times10^{-4}\%$ in any round, and in one round a raw pair had to be repaired after coming back slightly inverted. Figure~\ref{fig:colgen} makes clear why the lower bound needs the omitted settings. The pool-wide bound is $0.058$ hartree at the start and $0.089$ by round two, of little use until the largest price nears one. It is $0.41$ at round five and $0.51$ at round nine. One could close the active bracket every round and still have very little to certify if stopped early. Prices are not monotone either. Round three's top price is $8.9$, up from $7.7$ the prior round, since the act of admitting a setting alters the witness and can turn a moderate candidate into the worst violator. The test and stopping rules are laid out in Sections~\ref{sec:colgen} and \ref{sec:implementation}.

\begin{figure}[htbp]\centering
\includegraphics[width=0.75\linewidth]{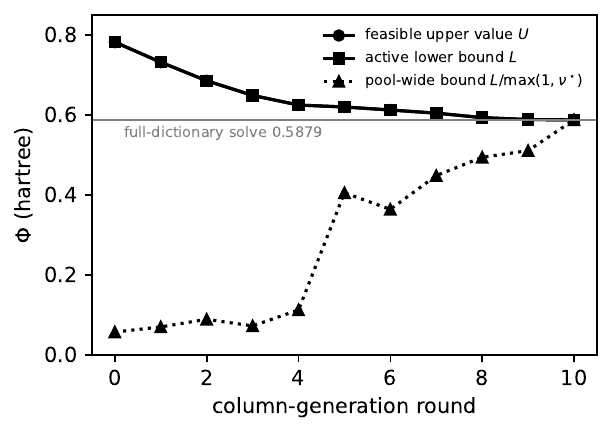}
\caption{The full-pool lower bound improves as fully commuting (FC) settings are added to 80 local LiH settings. The pool contains 20 FC candidates, with 2\% ground-state depolarization, $10^{-8}I$ covariance regularization, and unit shot costs. The active upper and lower values nearly agree, whereas the bound covering all candidates is $y^Th/\max(1,\nu_y^\star)$. Here $y$ is the checked witness, $h$ is the Hamiltonian-coefficient vector, and $\nu_y^\star$ is the largest candidate dual price. All candidates have been admitted by round ten, when the bounds meet at $0.5879$ hartree. The curves are generated from the stored run.}\label{fig:colgen}
\end{figure}

\subsection{Sampling costs in f-element Hamiltonians}\label{sec:felement}
In the course of the audited f-element work we have put five operator representations to the test for each of CeO, NdO, UO$_2^{+}$ and UO$_2^{2+}$. Prior to tapering the encodings call for 30 to 36 qubits, a figure that drops to between 27 and 32 once the process is complete. The block-invariant symmetry shift (BLISS)~\cite{Loaiza2023BLISS,Patel2025BLISS} is among the transformations employed, as it keeps the operator intact in the target sector. Table~\ref{tab:felement} provides the lower values $L$ for the measurement constant $\Phi$ along with the one-norms of the coefficients, $\lambda$. For all 20 of the stage calculations the protocol was the same: a Hartree-Fock reference with 2\% uniform noise, equal shot costs and standard rules for putting together finite QWC+FC collections. Their relative bracket gaps are below $0.02\%$.

\begin{table}[htbp]\centering\footnotesize\setlength{\tabcolsep}{4pt}
\caption{Lower bounds $L$ on $\Phi$ are reported for four f-element Hamiltonians, with coefficient one-norms $\lambda$ in parentheses. Both are in hartree, and the leading shot cost at standard error $\eps$ is $\Phi^2/\eps^2$. Collections combine qubit-wise commuting (QWC) and fully commuting (FC) settings, with at most 200 Hamiltonian-term scores per setting and 1500 FC candidates. Covariances use a Hartree-Fock determinant mixed with 2\% uniform noise, and shot costs are one. All 20 relative gaps $100(U/L-1)$ are below $0.02\%$. BLISS denotes the block-invariant symmetry shift, and tapering uses the reference determinant's symmetry sector to remove qubits.}
\label{tab:felement}
\begin{tabular}{lcccc}
\toprule
Stage & CeO & NdO & UO$_2^{+}$ & UO$_2^{2+}$ \\
\midrule
Qubits (untapered / tapered)  & 30 / 27         & 32 / 30        & 34 / 30         & 36 / 32 \\
Unshifted                     & 18.005 (186.79) & 7.992 (148.33) & 15.783 (131.83) & 37.742 (192.00) \\
Linear $N$ shift              & 17.489 (139.95) & 7.441 (109.39) & 15.617 (115.09) & 37.177 (162.05) \\
Full BLISS                    & 14.923 (65.70)  & 6.953 (52.67)  & 12.250 (65.44)  & 29.277 (98.13) \\
BLISS with symmetry shifts    & 14.966 (65.53)  & 6.972 (52.62)  & 12.325 (65.26)  & 29.267 (97.57) \\
Full BLISS, tapered           & 12.411 (65.70)  & 6.664 (52.67)  & 10.679 (65.44)  & 26.573 (98.13) \\
\bottomrule
\end{tabular}
\end{table}

The models show that symmetry shifting has a greater impact on $\lambda$ than on $\Phi$. A shift linear in particle number will bring $\lambda$ down by 13--26 per cent but only 1--7 per cent for $\Phi$. The change is mostly in the diagonal terms which are small in their fluctuations around the Hartree-Fock reference and can be read in concert. Relative to the unshifted operator, full BLISS reduces $\lambda$ by 49--65\% and $\Phi$ by some 13--22\%. From the squared ratios of the bounds one can infer leading shot reductions in the order of 24--40\%. Any further symmetry shifts amount to no more than a 0.7\% alteration of $\Phi$, with minor increases in three of the systems.

Tapering, when compared with full BLISS before the procedure, trims $\Phi$ by roughly 4--17\% without affecting $\lambda$. It takes four qubits from the uranyl ions, two from NdO and three from CeO. The Pauli representation and the settings from the construction rules are altered, as are the reference occupations in the covariance calculation. The sampling change is a reflection of this. The squared ratios point to an approximate 8--31\% reduction in leading shots.

\begin{table}[htbp]\centering\footnotesize\setlength{\tabcolsep}{4pt}
\caption{Lower values $L$, in hartree, are compared after the full block-invariant symmetry shift (BLISS), using qubit-wise commuting (QWC) settings alone and their union with fully commuting (FC) settings. All models use unit shot costs, a Hartree-Fock determinant mixed with 2\% uniform noise, at most 200 Hamiltonian-term scores per setting, and 1500 FC candidates. The final rows report $100(1-L_{\rm QWC+FC}/L_{\rm QWC})$, the reduction in the displayed measurement constant. Leading shot comparisons require squared ratios, and guaranteed savings require paired upper and lower bounds as in Eq.~\eqref{eq:paired-saving}.}
\label{tab:felement_dict}
\begin{tabular}{lcccc}
\toprule
                              & CeO    & NdO    & UO$_2^{+}$ & UO$_2^{2+}$ \\
\midrule
Untapered, QWC                & 22.047 & 8.225  & 26.131     & 45.305 \\
Untapered, QWC+FC             & 14.923 & 6.953  & 12.250     & 29.277 \\
Tapered, QWC                  & 17.864 & 8.116  & 21.153     & 43.436 \\
Tapered, QWC+FC               & 12.411 & 6.664  & 10.679     & 26.573 \\
Reduction in $L$, untapered (\%) & 32.3 & 15.5 & 53.1 & 35.4 \\
Reduction in $L$, tapered (\%)   & 30.5 & 17.9 & 49.5 & 38.8 \\
\bottomrule
\end{tabular}
\end{table}

Every f-element operator we have tested sees its measurement constant come down with FC settings. The reductions are 15.5--53.1\% before tapering and 17.9--49.5\% after it, with UO$_2^{+}$ giving the largest reduction in both cases. Since the numerical brackets are for all intents and purposes closed, the stage ratios are a fair estimate of the change in $\Phi$. To be certain of a minimum saving, however, one must combine the upper and lower bounds per Eq.~\eqref{eq:paired-saving}. In doing so, the upper values from the QWC+FC side and the lower ones from QWC certify savings of at least 54\%, 28\%, 78\% and 58\% for CeO, NdO, UO$_2^{+}$ and UO$_2^{2+}$ respectively prior to tapering. After tapering those guarantees are at least 51\%, 32\%, 74\% and 62\%. These are valid for the given collections, reference state and covariance model.

One can find the qubit-wise commuting costs for each compression stage in Supplementary Table S8. When looking at qubit-wise commuting settings in isolation the shifts are of less consequence: full BLISS brings the cost down 0.9--10.9\% over the unshifted operator and tapering by 1.3--19\%. Varying the depolarization to 1 or 5 per cent has little bearing on the full-BLISS QWC+FC costs (a 0.2--3.2\% difference at most) and does not disturb the ordering of the four systems (Supplementary Table S9), meaning the comparisons are independent of the proxy's noise level.

Such results from the proxy do not establish a trend with electronic correlation. The Hartree-Fock reference is uncorrelated until the depolarization is applied. To make claims about the prepared state's correlation would demand correlated-state covariance data. We have not included any such replacement in the numbers presented, though a sensitivity test could be run with a converged statevector from the adaptive derivative-assembled pseudo-Trotter variational quantum eigensolver (ADAPT-VQE) for the 27-qubit tapered CeO operator.
\section{Discussion}\label{sec:discussion}
When the upper and lower bounds are in near coincidence there is not much left for scheduling optimization to achieve. A wider bracket quantifies the optimization that remains unresolved. At that point a reallocation of shots or coefficients within the model yields only marginal returns and something more in the way of information or a new measurement procedure is called for.

Then there is shared-score pricing, which makes use of the information in existing outcomes which can improve a schedule at no extra setting. The H$_4$ calculations bear this out at fixed settings and state model. Even if the added observables' coefficients are cancelled from the Hamiltonian, their covariances will dampen the fluctuations in the fragments. Consequently any lower bound for a circuit family has to be explicit on what outcome processing is allowed.

RANGE is suited to the tasks of selecting and discovering settings, which are not necessarily convex in the way that coefficient sharing and shot allocation are. It is not meant to be a substitute for the conic solve, a division our control calculation in Sec.~\ref{sec:dictdesign} upholds for this implementation without ranking optimization algorithms in general.

The companion paper~\cite{Zahariev2026Clifford} looks at the problem from a structural vantage point. There the integer objective is a count of the readable fragments at the settings, with semidefinite graph witnesses bounding the Clifford grouping and the correlation rank acting as an exact obstruction for a specified orbital class in a given fermionic sector. Our weighted objective is different in that it allows for overlapping coefficients and factors in the per-shot cost and covariance. The $\Phi$ that results is the gauge function of the convex hull of the normalized setting balls, and while one model connects the two objectives, it is generally distinct from the unweighted fractional relaxation of the setting count. In the case of a pure two-body correlation block with unit costs in the maximally mixed fixed $(1,1)$ sector, the orbital gauge is the nuclear norm. The Bell-diagonal Heisenberg target for such a block exhibits an exact threefold leading-shot separation between the Clifford and orbital classes.

One can find a related application in Hadamard-test estimates of dynamical correlation functions~\cite{Zahariev2026OpenBoundary}, as the probe-operator decompositions there give rise to stratified sampling issues. While the present work does not put a price on a full dynamical algorithm, it provides a certificate for the given model in place of a generic measurement bound.

There are different operator properties at play when it comes to coefficient one-norm reduction and sampling optimization. The former is a sum of absolute coefficients. The latter is driven by state-specific correlations and fluctuations. A symmetry shift to take out diagonal weight might have little statistical impact if those terms are already weakly fluctuating around the Hartree-Fock reference. Tapering will alter the Pauli representation and any settings derived from it, so its sampling benefit depends on the reduced operator's measurement model.

The guarantee is subject to four constraints. For one, a proxy covariance is what certifies the proxy, with uncertainty in the prepared state manifesting in the score covariances. One can put bounds on this with preliminary measurements, but the concentration estimate employed here is conservative where there are many scores. More precise covariance data would tighten the cost bracket. Optimization cannot correct a misspecified statistical model, and numerical regularization of the covariance must be kept separate from the physical kind. A second point is that a witness cannot be extended to an unexamined family of measurements without verifying that any left-out setting meets its inequalities, whether by an upper bound on the largest violation or an exhaustive check. The absence of a violation in a heuristic is not a bound on what has not been seen. Third, the direct conic formulation runs large. The f-element instances we have include some $10^4$ settings and up to $2\times10^6$ variables, with roughly $1.5$--$3.3\times10^4$ terms including intermediate compression stages, all of which go to the Splitting Conic Solver (SCS) via the sparse equality block and block-diagonal cones. Pruning could push the formulation to $10^5$ terms, but this has yet to be shown. Finally, our estimator class is linear, unbiased and stratified for a single target. We do not optimize over every adaptive policy or channel. The merit of multi-observable, adaptive or informationally complete procedures is determined solely by the estimator they yield for the target.

These results on retained-setting enrichment make a case for more thorough classical processing. Its improvement shows that Hamiltonian-term-only bounds can leave substantial information untapped in the same outcomes. It is advisable to itemize the score basis, covariance prescription and setting collection, as well as the numerical residuals, the execution costs, $L$ and $U$, the dual witness and schedule, and the family over which prices were bounded or searched. This makes each bracket independently verifiable and allows the error versus shot count to be measured against the statistical model. One can tell from such records whether a method is putting known settings to better use, or if the baseline is close to optimum.

\section{Methods}\label{sec:methods}
\subsection{Observable coordinates and score statistics}
\label{method:coordinates}
We restrict the optimization to the part of the Hamiltonian that is fluctuating. Given a finite-dimensional Hermitian $H=cI+H'$, the variance is unaffected by estimating $\langle H'\rangle$ and then putting the known $c$ back in. Distinct Pauli strings are orthogonal under the trace pairing $\Tr(AB)$. Every admissible score must have a mean that matches the expectation of the observable to which it is assigned.

The observables at a setting determine the fragments it can accommodate. Let us denote them $B_{\ell1},\ldots,B_{\ell m_\ell}$ with real span $V_\ell$. Coefficients $x_{\ell a}$ specify
\begin{equation}
F_\ell(x_\ell)=\sum_{a=1}^{m_\ell}x_{\ell a}B_{\ell a}.
\end{equation}
We select a coordinate basis in common for the Hamiltonian and all observables. Here $h$ is the target and $A_\ell x_\ell$ the fragment at setting $\ell$. For a Pauli score, $A_\ell$ puts the coefficient in the string's coordinate. If a string is readable at more than one setting, the multiple copies of the coefficient will have to add up to the target.

Scores taken in the same shot can be correlated. With $S_{\ell a}$ having means $\Tr(\rho B_{\ell a})$, we put
\begin{equation}
(\Gamma_\ell)_{ab}=\Cov_\rho(S_{\ell a},S_{\ell b}).
\label{eq:covariance}
\end{equation}
For the case of joint eigenvalue measurements on commuting observables, the expression is
\begin{equation}
(\Gamma_\ell)_{ab}=\Tr(\rho B_{\ell a}B_{\ell b})
-\Tr(\rho B_{\ell a})\Tr(\rho B_{\ell b}).
\end{equation}
The standard deviation of the fragment reads
\begin{equation}
\norm{F_\ell(x_\ell)}_\rho=\sqrt{x_\ell^T\Gamma_\ell x_\ell}.
\end{equation}
The standard deviation is homogeneous and satisfies the triangle inequality, but a nonzero observable can be deterministic for a particular state. The standard deviation is thus a seminorm. We remove any deterministic contributions that are already known, though Eq.~\eqref{eq:dual_seminorm} retains the range condition needed for other zero-variance directions.

This approach extends to general quantum measurements via their score means and covariances. Consider a positive operator-valued measure: it is a set of positive outcome operators that add up to identity, with traces against $\rho$ yielding the probabilities. In a randomized ensemble one may even incorporate the setting choice into the outcome. The observables so assigned do not have to commute and the score statistics are still given by Eq.~\eqref{eq:covariance}.

One can assign a per-shot cost $c_\ell$ for an additive resource such as execution time. A value of one simply tallies the repetitions. A finite set of costs, scores and settings is denoted $\mathcal D$.

\subsection{Molecular inputs and covariance prescriptions}\label{sec:molecular}
Our H$_4$ models make use of all four spatial orbitals from restricted Hartree-Fock (RHF) on neutral four-electron singlet chains. With STO-3G at uniform 0.90 and 1.80 Å spacings we have eight Jordan-Wigner qubits. Discarding the identity and anything under $10^{-11}$ hartree in coefficient leaves 184 terms per spacing. The ground state is found by diagonalizing in the four-electron sector and its covariance is then evaluated with 2\% depolarization and no further regularization.

An H$_4$ optimization record will identify the Hamiltonian and covariance state along with RHF coefficients, analytic hydrogen STO-3G integrals, signed Pauli coefficients, statevectors, the score definitions and both primal and dual vectors. The checker is there to verify score readability after reconstructing the Hamiltonian and covariance. For the 69-setting models we rely on deterministic local settings. The complete family is defined by exhaustive axis assignments. Ten FC settings are added by sorted insertion at 0.90 Å and nine at 1.80 Å. The small-system validation solver rescales variables in accordance with their covariance to minimize the sum of group norms, while Clarabel is employed for the conic problems of the complete family.

We also have other small-molecule models based on RHF Hamiltonians with Jordan-Wigner encoding in STO-3G, and an H$_2$ run in the split-valence 6-31G Gaussian basis. Core orbitals are frozen for the LiH, BeH$_2$ and H$_2$O comparisons. The covariance is taken as $(1-p)|\psi_0\rangle\langle\psi_0|+pI/2^n$, where $\psi_0$ is the exact ground state of the operator in question and $p$ is typically $0.02$. To this we add $10^{-8}I$ by way of positive regularization, even in the complete-family LiH case. Since this cannot depress the primal objective, the lower bounds we report are for the regularized prescription only.

In the matter of schedules, each is optimized within its own settings first. Local candidates get one Pauli axis for every qubit. FC candidates come from sorted insertion and greedy extension to maximal commuting groups of Hamiltonian terms. Our comparisons cover QWC and FC groups using covariance-independent sorted insertion~\cite{Crawford2021Efficient}, largest-degree-first for FC~\cite{Verteletskyi2020MCC} and weight-greedy local bases as put forward for derandomized measurements~\cite{Huang2021Derandomized}.

An 18-qubit STO-3G Hamiltonian with 6891 nonidentity terms is the basis for the CH$_4$ benchmark. In the absence of depolarization (before the 2\% mark), one can obtain analytic Pauli expectations of $0$ or $\pm1$ from the Hartree-Fock determinant. The operator and covariance are common to both QWC and QWC+FC, which are compared here on the understanding that there are unit shot costs and a cap of 80 readable target terms for any given setting.

\subsubsection{Audited f-element operators}\label{sec:felement_methods}
Inputs for the f-element operators come from state-averaged complete active space self-consistent field (CASSCF) work with frozen natural orbitals. We have CeO, NdO, UO$_2^{+}$ and UO$_2^{2+}$ holding 15, 16, 17 and 18 spatial orbitals in turn, put in interleaved Jordan-Wigner order. A check of the two-body integrals shows that each shifted operator holds the unshifted reference-determinant energy to within $10^{-10}$ hartree.

The transformations we compare do not alter the operator in the physical sector of interest. There is the linear particle-number shift, and full BLISS with its number and number-squared operators as well as products of the particle-number deviation and spin-summed one-body operators. There is also BLISS augmented by the spin and mixed symmetry shifts specified in the audit. By way of independent commuting $\mathbb Z_2$ symmetries (eigenvalues $+1$ and $-1$) we can taper full BLISS and remove one qubit per symmetry. Deleting the pivot qubits leaves reduced reference occupations whose energies are indistinguishable from the untapered ones at every printed digit. This serves as an energy test to back up the coefficient-level audit.

There is no need to keep a many-electron statevector on hand to get the f-element covariances. We apply the CH$_4$ benchmark's depolarized Hartree-Fock approach, which puts a limit of 200 readable target terms per setting and 1500 FC candidates. Solver tolerances run at $10^{-6}$ for full BLISS, additional shifts and tapered operators, and $10^{-5}$ for the unshifted and linear-shift types. On a single Frontier node the 20 stage calculations in Table~\ref{tab:felement} take between four and 24 minutes, with relative gaps under $0.02\%$. Table~\ref{tab:felement_dict} comes from some extra QWC-only runs. The optimization certificates are for the numerical Hamiltonians provided. Separate audits tell us how well those stand in for the molecular models in question.

To put the numerical side to the test we have twelve synthetic instances, fully specified. These cover diagonal and correlated overlap as well as four-qubit local Pauli on random product states, and they confirm the reconstruction, all dual inequalities, both objectives and the final gap.
\subsection{Optimal stratified allocation}
\label{app:proofs}
Take a splitting $h=\sum_\ell A_\ell x_\ell$ that will reconstruct the target and let the fragment variances be
\begin{equation}
v_\ell=x_\ell^T\Gamma_\ell x_\ell\geq0.
\end{equation}
Assuming the shot batches are independent, the variance of the sum of sample means is
\begin{equation}
\sum_\ell\frac{v_\ell}{N_\ell}.
\end{equation}
We make the shot allocations continuous and solve
\begin{equation}
\min_{N_\ell>0}\sum_\ell c_\ell N_\ell
\quad\text{subject to}\quad
\sum_\ell\frac{v_\ell}{N_\ell}\leq\eps^2.
\end{equation}
The Lagrangian then is
\begin{equation}
\mathcal L(N,\lambda)=\sum_\ell c_\ell N_\ell
+\lambda\left(\sum_\ell\frac{v_\ell}{N_\ell}-\eps^2\right).
\end{equation}
For fragments with positive variance, stationarity requires
\begin{equation}
c_\ell-\lambda\frac{v_\ell}{N_\ell^2}=0,
\qquad N_\ell=\sqrt{\lambda v_\ell/c_\ell}.
\end{equation}
At the optimum the error constraint is active, meaning
\begin{equation}
\frac1{\sqrt\lambda}\sum_\ell\sqrt{c_\ell v_\ell}=\eps^2,
\end{equation}
so that
\begin{equation}
\sqrt\lambda=\frac{\sum_\ell\sqrt{c_\ell v_\ell}}{\eps^2}.
\end{equation}
Putting this in the cost we have
\begin{equation}
\sum_\ell c_\ell N_\ell
=\sqrt\lambda\sum_\ell\sqrt{c_\ell v_\ell}
=\frac{\left(\sum_\ell\sqrt{c_\ell v_\ell}\right)^2}{\eps^2}.
\end{equation}
Zero-variance fragments are taken by continuity and carry no leading sampling cost. One minimizes the numerator across all reconstructing splittings to arrive at Eq.~\eqref{eq:primal}. With a finite set of settings, rounding up any positive allocation only makes a difference in lower-order terms as $\eps\to0$. A preliminary readout will do to fix a deterministic contribution if it is not already in hand.

\subsection{Derivation of the dual constraints}
\label{app:dual}
The primal problem may be written
\begin{equation}
\min_x\sum_\ell\varphi_\ell(x_\ell)
\quad\text{subject to}\quad\sum_\ell A_\ell x_\ell=h,
\end{equation}
where
\begin{equation}
\varphi_\ell(x)=\sqrt{c_\ell}\sqrt{x^T\Gamma_\ell x}.
\end{equation}
and the Lagrangian is
\begin{equation}
\mathcal L(x,y)=\sum_\ell\varphi_\ell(x_\ell)
+y^T\left(h-\sum_\ell A_\ell x_\ell\right).
\end{equation}
Homogeneity dictates that the infimum over a setting variable is finite if and only if
\begin{equation}
(A_\ell^Ty)^Tx\leq\varphi_\ell(x)\qquad\text{for every }x.
\end{equation}
One cannot scale a violating direction without sending the Lagrangian to negative infinity. The condition for a finite infimum is therefore
\begin{equation}
\norm{A_\ell^Ty}_{\Gamma_\ell,*}\leq\sqrt{c_\ell}.
\end{equation}
This is the condition for a finite infimum. With it, the setting contributions are zero at the infimum and we are left with $y^Th$ as the objective, i.e. Eq.~\eqref{eq:dual}. Weak duality is enough for a feasible lower witness. Once deterministic directions are quotiented out, the relative-interior assumption gives us strong duality and attainment.

\subsubsection{A certificate verifiable by hand}
\label{sec:example}
Suppose $h=(1,2,1)^T$ and we have two settings with unit costs and $\Gamma_1=\Gamma_2=I$, the readable spaces being $e_1,e_2$ and $e_2,e_3$. An even split of the second coordinate yields $x_1=x_2=(1,1)^T$. It is a valid reconstruction of $h$ ($A_1x_1+A_2x_2=h$) with an objective $U=2\sqrt2$. The dual $y=(1,1,1)^T/\sqrt2$ meets the constraints for both settings, $\norm{A_1^Ty}_2=\norm{A_2^Ty}_2=1$, and $L=2\sqrt2$. Since the values match, there is no schedule to be had from these two settings that is better. The logic extends to larger certificates with more coordinates.
\subsection{Proof of the shared-score criterion}\label{method:score-proof}
\begin{proof}
The block covariance is positive semidefinite, so $g_\ell$ is in the range of $\Gamma_\ell$ and $q_\ell\geq0$. Feasibility of the original witness means $z_\ell$ is in $\range\Gamma_\ell$ and $z_\ell^T\Gamma_\ell^\dagger z_\ell\leq c_\ell$. When $q_\ell>0$ the inverse-covariance quadratic form for the score vector becomes
\[
z_\ell^T\Gamma_\ell^\dagger z_\ell+
\frac{(t-g_\ell^T\Gamma_\ell^\dagger z_\ell)^2}{q_\ell}.
\]
Holding this to $c_\ell$ confines $t$ to $[m_\ell-r_\ell,\,m_\ell+r_\ell]$. Should $q_\ell$ be zero, the interval is just the point $t=m_\ell$.

All of these intervals must contain the shared dual coordinate. They intersect, as in Eq.~\eqref{eq:ghost-interval-test}, provided the largest lower bound does not pass the smallest upper one. The new coordinate has a zero target coefficient so $y^Th$ is unchanged. Nor does the schedule change if we give the added score zero primal coefficients. Thus a previously closed bracket remains closed at the same value.
\end{proof}
\subsection{Cost certificates under covariance uncertainty}
\label{sec:robust}

Confidence bounds on the score covariances of the prepared state are what we get from the initial measurements. Our approach is to put together lower and upper matrices that enclose all setting covariances with a probability of no less than $1-\delta$:
\begin{equation}
\Pr\!\left(\underline\Gamma_\ell\preceq\Gamma_\ell(\rho)\preceq\overline\Gamma_\ell
\text{ for every }\ell\right)\geq1-\delta.
\label{eq:coverage}
\end{equation}
In what follows, $M\preceq M'$ denotes that the difference $M'-M$ is positive semidefinite. In other words, $x^TMx\leq x^TM'x$ for any real vector $x$. These covariance inequalities
\begin{equation}
\underline\Gamma_\ell\preceq\Gamma_\ell(\rho)\preceq\overline\Gamma_\ell
\label{eq:sandwich}
\end{equation}
thus place a bound on every fragment variance. The upper matrices give us a conservative schedule cost while the lower ones provide a legitimate floor for the optimum.

\begin{theorem}[Simultaneous covariance coverage from finite samples]
\label{thm:coverage}
Take $M$ fixed settings where the Pauli observables $P^{(\ell)}_1,\ldots,P^{(\ell)}_{m_\ell}$ are declared at setting $\ell$. A shot will yield a vector $v_i^{(\ell)}$ of $m_\ell$ eigenvalues from $\{+1,-1\}$. The shots $i=1,\ldots,N_\ell$ for each setting are independent and identically distributed. Let
\begin{align*}
\widehat\Gamma_\ell&=\frac1{N_\ell}\sum_{i=1}^{N_\ell}v_i^{(\ell)}v_i^{(\ell)\top}
-\bar v^{(\ell)}\bar v^{(\ell)\top},\\
\bar v^{(\ell)}&=\frac1{N_\ell}\sum_i v_i^{(\ell)},
\end{align*}
and set
\begin{align*}
\kappa_\ell&=\log(2m_\ell)+\log(2M/\delta),\\
r_\ell(\delta)&=4m_\ell\sqrt{\frac{\kappa_\ell}{N_\ell}}
+\frac32\frac{m_\ell\kappa_\ell}{N_\ell}.
\end{align*}
With probability at least $1-\delta$ we have the following hold for all settings at once:
\[
\bigl\|\widehat\Gamma_\ell-\Gamma_\ell\bigr\|_{\rm op}\leq r_\ell(\delta).
\]
\end{theorem}

\begin{proof}
We can drop the index by fixing a setting. Put $m=m_\ell$, $N=N_\ell$, $\mu=\E v_i$ and $S=\E v_iv_i^\top$. Since an outcome is an eigenvalue in $\{+1,-1\}$, $\|v_i\|_2^2=m$.

For the second-moment error, define $X_i=v_iv_i^\top-S$. These are symmetric, centered and independent matrices with $\|X_i\|_{\rm op}\leq2m$. The variance parameter is such that
\[
\left\|\sum_{i=1}^N\E X_i^2\right\|_{\rm op}
=N\|mS-S^2\|_{\rm op}\leq\frac{Nm^2}{4}.
\]
Every eigenvalue of $S$ is in $[0,m]$ and $\lambda(m-\lambda)\leq m^2/4$, the quadratic dependence on the number of scores being reached at $\lambda=m/2$. Applying the matrix Bernstein inequality~\cite{Tropp2012} one obtains
\[
\Pr\!\left(\|\widehat S-S\|_{\rm op}\geq t\right)
\leq2m\exp\!\left(\frac{-Nt^2/2}{m^2/4+2mt/3}\right),
\]
where $\widehat S=N^{-1}\sum_i v_iv_i^\top$. With $\kappa_\ell=\log(4m/\delta_\ell)$ it is the case that, with probability at least $1-\delta_\ell/2$,
\[
\|\widehat S-S\|_{\rm op}\leq
m\sqrt{\frac{\kappa_\ell}{2N}}+\frac43\frac{m\kappa_\ell}{N}
=: \rho_\ell.
\]

As for mean error, Hoeffding's inequality tells us $\Pr(|\bar v_j-\mu_j|\geq t)\leq2\exp(-Nt^2/2)$ since each coordinate is in an interval of width 2. Over the $m$ coordinates a union bound gives
\[
\|\bar v-\mu\|_\infty\leq\sqrt{\frac{2\kappa_\ell}{N}},
\qquad\|\bar v-\mu\|_2\leq\sqrt{\frac{2m\kappa_\ell}{N}}
\]
with probability at least $1-\delta_\ell/2$. Given $\|\bar v\|_2,\|\mu\|_2\leq\sqrt m$ this means
\[
\begin{aligned}
\|\bar v\bar v^\top-\mu\mu^\top\|_{\rm op}
&\leq\|\bar v-\mu\|_2(\|\bar v\|_2+\|\mu\|_2)\\
&\leq2\sqrt2\,m\sqrt{\kappa_\ell/N}.
\end{aligned}
\]

Both are satisfied with probability at least $1-\delta_\ell$. On their joint event we have
\[
\begin{aligned}
\|\widehat\Gamma_\ell-\Gamma_\ell\|_{\rm op}
&\leq\rho_\ell+2\sqrt2\,m\sqrt{\kappa_\ell/N}\\
&=\left(\frac1{\sqrt2}+2\sqrt2\right)m\sqrt{\kappa_\ell/N}
+\frac43m\kappa_\ell/N\\
&=\frac5{\sqrt2}m\sqrt{\kappa_\ell/N}+\frac43m\kappa_\ell/N\\
&\leq r_\ell(\delta).
\end{aligned}
\]
using $5/\sqrt2<4$ and $4/3<3/2$ in the last step. Setting $\delta_\ell=\delta/M$ and taking a union bound over the settings yields $\kappa_\ell=\log(4m_\ell M/\delta)$ and a failure probability not exceeding $\delta$.
\end{proof}

One finds the bound to be more conservative with an increase in the number of scores per setting. The leading radius is linear in $m_\ell$ but only logarithmic in the number of settings. To halve the radius, which scales as $N_\ell^{-1/2}$, calls for some four times as many preliminary shots. Supplementary Note S3 has the numerical intervals for this proved radius. Being intentionally conservative, the linear width can result in wide intervals even when the preliminary budget is large.

We have also looked at how much tighter the radius might be for the ensembles under test. RANGE~\cite{RANGE} was used to search $r_\ell(c,a)=c\,m_\ell^a\sqrt{\kappa_\ell/N_\ell}$ for the continuous parameters $(c,a)$ that would minimize the mean radius while maintaining at least $1-\delta$ empirical simultaneous coverage in our simulations. For the twelve LiH settings (2--55 scores apiece) the best parameters found give
\[
r_\ell\approx0.67\,m_\ell^{0.81}\sqrt{\kappa_\ell/N_\ell}.
\]
In those tests the calibrated intervals show at least 95\% empirical simultaneous coverage and are 10.4--10.6 times more compact than what is proved. The fitted prefactor is roughly six times down and the width exponent is close to 0.8 rather than 1.

Such fitted quantities are descriptive of the states and settings sampled. They do not constitute a new concentration theorem or a guarantee that can be transferred. There is no need for the chemical models to come up against the worst-case matrix-variance of $m^2/4$ assumed in the proof, and any sharper result would have to be proven on its own. Consequently, all robust certificates we report make use of the proved radius.

In practice we employ
\begin{align}
a_\ell&=\begin{cases}
1-r_\ell/\lambda_{\min}(\widehat\Gamma_\ell),
&\lambda_{\min}(\widehat\Gamma_\ell)>r_\ell,\\
0,&\lambda_{\min}(\widehat\Gamma_\ell)\leq r_\ell,
\end{cases}\\
\underline\Gamma_\ell&=a_\ell\widehat\Gamma_\ell,
\qquad\overline\Gamma_\ell=\widehat\Gamma_\ell+r_\ell I.
\label{eq:empirical-sandwich}
\end{align}
Should the smallest covariance eigenvalue prove unresolvable, the construction zeros out the lower matrix. The lower bound on cost may be of little information then, but it is still valid. Singular matrices are allowed, under the same dual range condition as the exact-covariance problem.

\begin{theorem}[A cost bracket with simultaneous confidence bounds]
\label{thm:robust}
Assume the covariance inequalities of Eq.~\eqref{eq:sandwich} are in force with probability at least $1-\delta$. Define
\begin{equation}
U_{\rm rob}=\min_{\sum A_\ell x_\ell=h}
\sum_\ell\sqrt{c_\ell}\sqrt{x_\ell^T\overline\Gamma_\ell x_\ell}
\label{eq:robust_upper}
\end{equation}
and
\begin{equation}
\begin{aligned}
L_{\rm rob}=\max_y\quad&y^Th\\
\text{s.t.}\quad&\norm{A_\ell^Ty}_{\underline\Gamma_\ell,*}\leq\sqrt{c_\ell},
\quad\forall\ell.
\end{aligned}
\label{eq:robust_lower}
\end{equation}
Then with probability at least $1-\delta$,
\begin{equation}
\boxed{L_{\rm rob}\leq\Phi\leq U_{\rm rob}.}
\label{eq:robust_cert}
\end{equation}
The allocation based on the upper covariances will have a leading cost of no more than $U_{\rm rob}^2/\eps^2$ on this event, and any permissible schedule will have leading cost at least $L_{\rm rob}^2/\eps^2$. This covers the covariance uncertainty for later measurement batches, though the shot cost of the preliminary work is an extra item to be factored into the experimental budget.
\end{theorem}

\subsection{Why covariance bounds preserve the certificate}
\label{app:robust}
Conditioning on simultaneous covariance coverage,
\begin{equation}
0\preceq\underline\Gamma_\ell\preceq\Gamma_\ell\preceq\overline\Gamma_\ell.
\end{equation}
we see that for any coefficient vector $x^T\Gamma_\ell x\leq x^T\overline\Gamma_\ell x$. An evaluation of a fragment using the upper covariance is therefore an upper bound on the true variance. The minimization of that conservative objective provides the upper cost and the shot allocation.

Turning to the dual, $\underline\Gamma_\ell\preceq\Gamma_\ell$ means the lower matrix's range is within that of the true covariance, so on the smaller range
\begin{equation}
z^T\underline\Gamma_\ell^\dagger z\geq z^T\Gamma_\ell^\dagger z.
\end{equation}
Any witness that is feasible for the lower covariance matrices will do for the true ones and its objective will be a lower bound on the true optimum. This, along with the primal side of things, establishes Theorem~\ref{thm:robust} on the simultaneous-coverage event.

The matrix-concentration method of~\cite{Tropp2015} is not the only option. Any estimator will do as long as it can be shown to deliver the necessary simultaneous coverage for the input matrices.
\subsection{Testing additional settings}
\label{sec:colgen}
The dual witness is what we use to vet candidate settings prior to tackling an enlarged optimization problem. Take $\mathcal D$ as the optimized finite set and $\mathcal L$ for the broader candidate family. For any given setting $V$ in the latter, one computes
\begin{equation}
\nu_y(V)=\frac{\norm{A_V^Ty}_{\Gamma_V,*}}{\sqrt{c(V)}}.
\label{eq:pricing}
\end{equation}
and we refer to $\nu_y(V)$ as the \emph{dual price}. Should this price be at or below one, the setting is in compliance with the present witness constraint. If the upper and lower bounds are one and the same, there is no point in adding it as the optimum will not be affected.

Conversely, a price in excess of one is a violation of the current constraint and flags the setting for reoptimization. We then add it and recalculate the cost bracket. The fact of the violation does not in itself mean the cost has been reduced since some other witness for the old optimum might still be feasible. This cycle of reoptimization after adding settings is known as \emph{column generation}.

To put an upper bound on the largest candidate price and extend the certificate to all of $\mathcal L$, we define
\begin{equation}
\nu_y^\star=\sup_{V\in\mathcal L}\nu_y(V).
\label{eq:pricing_sup}
\end{equation}
In the event that $\nu_y^\star\leq1$, $y$ is feasible across the whole of $\mathcal L$ and a closed active bracket is proof of optimality for the larger family.

\begin{proposition}[A bracket for a larger setting family]
\label{prop:colgen}
With $U_\D$ as the objective value for a feasible schedule on $\mathcal D$ and $y$ a corresponding witness, assume a valid bound can be put on $\mathcal L$ such that $\nu_y^\star\leq1+\eta$ for some $\eta\geq0$. Then
\begin{equation}
\frac{y^Th}{1+\eta}\leq\Phi^{\mathcal L}
\leq\Phi^{\mathcal D}\leq U_\D.
\label{eq:colgen_cert}
\end{equation}
The superscripts serve to tell the two families apart. When $\eta$ is zero and $U_\D=y^Th$, the active schedule is optimal over the entirety of $\mathcal L$.
\end{proposition}

Even if the top price is above one, a finite upper bound limits the unresolved improvement in sampling cost. Provided $y^Th>0$, the measurement constant of the schedule will be within a factor $(1+\eta)U_\D/(y^Th)$ of the optimum on $\mathcal L$, with the leading-cost factor being the square of the ratio. One can check every price directly from an enumerated pool, but for an implicit family an upper bound has to be derived from a relaxation or branch-and-bound.

Heuristics are of service in putting forward candidates. Our approach is to form initial pools by grouping Hamiltonian terms and perhaps employ RANGE to find further commuting subsets, using the exact dual price as a measure of fitness~\cite{RANGE}. In the case of LiH, an unseeded search produced a price of $7.9$ compared to roughly $13.6$ for the best in the pool. Thirty cycles of starting from the pool did nothing to better that figure, returning $13.5993$. While these tests locate the best candidate found, the largest price they uncover only gives a lower bound on what is possible. It cannot vouch for an unexamined family. An upper bound covering the whole family is needed for that.
\subsection{Removing redundant settings}\label{method:dominance}
\begin{lemma}[Redundant settings]\label{lem:dominance}
If every fragment readable at setting $\ell$ can also be read at setting $k$ and $c_\ell\geq c_k$, and shared scores have identical covariances at the two settings, as do eigenvalue scores for the same Pauli strings, then $\Phi$ is unaffected by the removal of setting $\ell$. Settings of equal cost and identical readable space are therefore equivalent.
\end{lemma}

\begin{proof}
Consider a feasible splitting and transfer $F_{\ell}\in V_{\ell}$ to $k$. Reconstruction is maintained as $F_{\ell}\in V_{k}$, and so is the variance of the fragment under our covariance assumption. The seminorm triangle inequality together with the shot cost ordering yields $\sqrt{c_{k}}\,\|F_{k}+F_{\ell}\|_{\rho}\le\sqrt{c_{k}}\,\|F_{k}\|_{\rho}+\sqrt{c_{\ell}}\,\|F_{\ell}\|_{\rho}$. It follows that the objective cannot be raised by the move. Since any feasible splitting is open to this, the optimum is left as is upon deletion of $\ell$.
\end{proof}

\subsection{A lower bound from depolarization}
\label{sec:depolfloor}
One can derive a lower bound that applies simultaneously to every permitted Pauli-score setting from the completely mixed part of a depolarized state.

\begin{proposition}[Depolarization bound]\label{prop:depol}
For $\rho=(1-p)|\psi\rangle\langle\psi|+pI/2^{n}$ where $0<p\le1$, and with each setting reading eigenvalue scores of distinct nonidentity Pauli strings, the covariance of every setting will be $\Gamma_{\ell}\succeq pI$. Assuming unit shot costs: (i) $\Phi\ge\sqrt{p}\,\|h\|_{2}$ for any setting collection and Pauli-score basis. This is true for Clifford or fermionic settings as long as their scores conform to the Pauli-score assumption, and it accommodates shared zero-target scores. (ii) Any setting price is bounded by $\nu_{y}(V)\le\|A_{V}^{T}y\|_{2}/\sqrt{p}\le\|y\|_{2}/\sqrt{p}$. So a witness with $\|y\|_{2}\le\sqrt{p}$ is feasible for the full range of Pauli-score settings.
\end{proposition}

\begin{proof}
The expectation for Pauli eigenvalue scores is $\mathbb E_{\rho}[S_{i}S_{j}]=(1-p)\langle\psi|P_{i}P_{j}|\psi\rangle+p\,\delta_{ij}$ with $\mu_{\rho}=(1-p)\mu_{\psi}$. Therefore,
\begin{equation}
\Gamma_{\rho}=(1-p)\Gamma_{\psi}+pI+p(1-p)\mu_{\psi}\mu_{\psi}^{T}\succeq pI .
\end{equation}
Hence for each fragment we have $\sqrt{x_{\ell}^{T}\Gamma_{\ell}x_{\ell}}\ge\sqrt{p}\,\|x_{\ell}\|_{2}$. As $A_{\ell}$ is a coordinate injection, the Euclidean norm is preserved and any feasible reconstruction must satisfy
\begin{equation}
\begin{split}
\sum_{\ell}\sqrt{x_{\ell}^{T}\Gamma_{\ell}x_{\ell}}&\;\ge\;\sqrt{p}\sum_{\ell}\|x_{\ell}\|_{2}\\
&\;\ge\;\sqrt{p}\Big\|\sum_{\ell}A_{\ell}x_{\ell}\Big\|_{2}=\sqrt{p}\,\|h\|_{2}.
\end{split}
\end{equation}
Minimizing over these proves the bound on $\Phi$. Zero-target coordinates do not alter the norm on the right. To get (ii), let $z=A_{V}^{T}y$ and we have $\nu_{y}(V)^{2}=z^{T}\Gamma_{V}^{\dagger}z\le\|z\|_{2}^{2}/p$. A projection of the coordinates will not put up the Euclidean norm, which accounts for the second inequality. In the case $p=1$, both inequalities become equalities since the covariance is $I$.
\end{proof}

This is independent of the orbital frame as $\|h\|_{2}=2^{-n/2}\,\|H-2^{-n}\operatorname{Tr}(H)\,I\|_{F}$ is unitarily invariant. One may define the Frobenius norm $\|\cdot\|_F$ as the square root of the sum of the squared magnitudes of a matrix's elements. Even though the covariance inequalities become equalities when $p=1$, there is no guarantee that the bound on $\Phi$ will be tight. Take $H=X+Z$ for instance: if one measures the $X$ and $Z$ terms in isolation, $\Phi$ comes to 2, while $\|h\|_2$ is only $\sqrt2$.

In the case of H$_4$ states with 2\% depolarization, this general bound is far from as strong as what is obtained with combined-family results. Where Sec.~\ref{sec:complete-family} has $0.901$ and $0.864$ at 0.90 and 1.80 Å, the bound yields $0.135$ and $0.072$ hartree. It is true that the bound encompasses Hamiltonian terms and other Pauli scores, but it tends to zero with vanishing depolarization. Its merit is in its scope. It is applicable to any setting and score basis in the Pauli-score class at once, unlike the certification procedure of Sec.~\ref{sec:colgen} which demands a finite pool. With the H$_4$ settings used here the smallest covariance eigenvalue is exactly $p$, so the covariance inequality is attained. To do better would call for some knowledge of the pure-state component.
\subsection{Numerical feasibility repair}
\label{app:repair}
It is only after an inspection of constraint residuals that numerical objective values can be considered upper or lower bounds. Solvers using floating-point arithmetic will typically honor constraints only to within certain tolerances. The correction below is meant to restore exact feasibility in real arithmetic and to put in place the quantities for the numerical checker.

\begin{proposition}[Repair of primal and dual feasibility]
\label{prop:repair}
With $A=[A_1\ \cdots\ A_M]$ and $h\in\range(A)$, let $B$ be a map for which $ABA=A$, making $AB$ the identity on $\range(A)$. The checker is to either build $B$ by some deterministic rule or confirm the property holds. Given a numerical splitting $x$ we set $r=h-Ax$ and
\[
x_{\rm cert}=x+Br,
\qquad U_{\rm cert}=\sum_\ell\sqrt{c_\ell}
\bigl\|\Gamma_\ell^{1/2}x_{{\rm cert},\ell}\bigr\|_2.
\]
The corrected version then has $Ax_{\rm cert}=h$ and the objective is bounded by $\Phi\leq U_{\rm cert}$. As for a dual vector $y$, one computes $z_\ell=A_\ell^Ty$ and verifies $z_\ell\in\range(\Gamma_\ell)$ across the board. Should all such checks hold, the sign of $y$ is chosen to make $h^Ty\geq0$ and we have

\[
\nu(y)=\max_\ell\frac{\sqrt{z_\ell^T\Gamma_\ell^\dagger z_\ell}}{\sqrt{c_\ell}},
\qquad y_{\rm cert}=\frac{y}{\max\{1,\nu(y)\}}.
\]
a dual feasible vector where $L_{\rm cert}=h^Ty_{\rm cert}\leq\Phi$. If a check does not pass, no amount of scaling will remove the invalid part. In that event $y_{\rm cert}$ is returned as $0$ along with the valid but uninformative $L_{\rm cert}=0$. In either scenario
\[
L_{\rm cert}\leq\Phi\leq U_{\rm cert}.
\]
and provided $L_{\rm cert}>0$ the gap is nonnegative.
\end{proposition}

For Pauli scores the condition $h\in\range(A)$ is met so long as each target term is readable in some setting. There is no need to compute a dense pseudoinverse of $A$ to effect a sparse repair of the reconstruction. One can simply apportion each residual coefficient to a readable copy. An explicit repair map or documented procedure is necessary for a general observable basis.

\begin{proof}
$r$ being in $\range(A)$ means $ABr=r$, whence $Ax_{\rm cert}=Ax+ABr=h$. This feasible point provides an upper bound on the primal objective. Homogeneity guarantees $\nu(y_{\rm cert})\leq1$ for a range-valid witness, so weak duality gives the lower bound. The zero vector serves as a fallback in any dual model.
\end{proof}

Using the covariance in place of the pseudoinverse can lead to spurious lower bounds. Let $A=h=c=1$ and $\Gamma=0.01$. The optimum is $\Phi=0.1$. A primal-norm test might accept $y=1$ since $0.1\leq1$, yielding the absurdity $1\leq\Phi\leq0.1$. The proper dual check with $\sqrt{\Gamma^\dagger}=10$ rescales $y$ to $0.1$ and puts the bracket in order. When $\Gamma=0$ the range condition dictates $z=0$.

This kind of sparse repair is linear in the number of coefficient copies and is suited to models of $10^4$ settings and $10^6$ variables. In our large QWC+FC work the reconstruction residuals are $\|r\|_2\leq2.6\times10^{-6}$ and $U_{\rm cert}-U\lesssim10^{-5}$. These are small figures but they matter, for a recomputed certificate is not the same as an approximately feasible solver output. One must recompute the objective from the corrected coefficients to be sure the upper bound is sound.

The checker will reconstruct both vectors and their objectives from the model on file. We have tested it against scalar false-bound cases, singular-covariance ranges of both types, and corrupted primal--dual pairs, as well as an independent optimization. In the corruption tests the errors of the repaired reconstruction are in the region of $10^{-15}$. A further test was run to repair $10^4$ variables eschewing dense factorization.

Every numerical certificate issued will note the residuals and tolerances post-correction. While the checker operates in floating-point without the use of directed rounding to enclose operations, the distinction is not trivial close to an optimum. An unrepaired H$_2$ calculation could see $U$ come in under $L$ by $10^{-7}$ or so.

\subsection{Required certificate records}
\label{app:schema}
There is no need to rerun the optimization to get the certificate from the stored model and vectors. The record will have $h$ and, per setting, $A_\ell$, $\Gamma_\ell$, $c_\ell$ and $x_\ell$. The dual $y$, the tolerances and the repair map $B$ (or the rule for its construction) are also included.

Physically, the coordinates must be defined: the state or covariance model, and the ordered observables, their signs and normalization. Where covariances are in question the upper matrices put the schedule to the test and the lower ones the witness. The output will show both objectives and the relevant residuals, as well as the squared bound ratio that keeps the excess leading cost in check.
\subsection{Optimization and independent verification}
\label{sec:implementation}
The declared measurement model is common to both the optimization and its independent verification. We put in place the settings, scores and shot costs before verifying that the observables we can read do in fact span the target. For covariances we rely on a classical reference, exact state calculations or preliminary measurements. Once Eq.~\eqref{eq:socp} has been solved the model is exported along with the primal and dual vectors. An independent checker then repairs the reconstruction, tests dual feasibility and notes the residuals and bounds.

There are limits to what a certificate can do. It applies only to the family of constraints under review, bounding the optimum of a fixed set. To extend it to an enumerated pool requires every candidate constraint be checked, while an unrestricted family is covered only if one can put a bound on all settings left out. One should not infer from the cost certification of a RANGE selection that another collection of equal size does not come in lower.

To put the statistical model to the test one might look at the experimental root-mean-square error (RMSE) versus shot count. Such a comparison with the stored bounds could help separate whether there is a covariance mismatch or if the issue lies with incomplete optimization or the measurement family itself.

A reproducible comparison reports, for each target, the empirical RMSE as a function of shots alongside $U$, $L$, $U/L-1$ and $\nu_y^\star-1$, making clear over which family the latter is bounded or sought. The methods used for this may be QWC or FC grouping, overlapping groups, deterministic or derandomized shadows, single-term Pauli estimation, ShadowGrouping, resource-optimized grouping shadows (ROGS), compatible fermionic fragments (Sec.~\ref{sec:relations}) or dual-frame estimators when the data are informationally complete. Published schedules are to be kept distinct from any reimplementation of a strategy class, and each method will have its own score space and covariance convention. An informative result is not necessarily a speedup. A certificate might simply show that a baseline is close to optimal or that the current collection is better served by a larger model. By reporting these alternatives the outcome remains interpretable regardless of the heuristic at work.

\subsubsection{Removing settings from a computed schedule}
\label{sec:lossless-support}
It is possible to rescale the fragments of a schedule and remove some settings at no extra cost. Let their coefficient columns be $f_\ell=A_\ell x_\ell$ and the weighted standard deviations $q_\ell=\sqrt{c_\ell x_\ell^T\Gamma_\ell x_\ell}$. We solve
\begin{equation}
\min_{t\geq0}\ \sum_\ell q_\ell t_\ell
\quad\text{subject to}\quad\sum_\ell t_\ell f_\ell=h.
\label{eq:lossless-support}
\end{equation}
Since $t_\ell=1$ gives back the original schedule, the optimal rescaling will not see an increase in cost. Homogeneity means the rescaled fragment $t_\ell x_\ell$ carries its scaled cost with it, and the standard support bound for a nonnegative linear program with equality constraints means a basic feasible optimum will not use more than $\rank[f_1\ \cdots\ f_M]$ settings. Feasibility of the dual witness is preserved when settings are deleted and so too is the lower bound. This alters the magnitude of the fragments but not their direction. Finding the smallest unrestricted collection is a separate matter as the directions are free to change. In a synthetic three-coordinate case, for instance, rescaling brought 80 settings down to two without affecting the cost.

\subsubsection{Numerical implementation and search records}
Large conic problems are built in sparse form directly rather than via a symbolic conversion. SCS is used for the block-diagonal cones and reconstruction matrix, with Clarabel for the H$_4$ complete-family work. The runtimes quoted for the f-element calculations in Sec.~\ref{sec:felement_methods} are elapsed times on one node.

The outer nonconvex problems are handled by RANGE (Robust Adaptive Nature-inspired Global Explorer~\cite{RANGE}). This involves selecting a subset of settings, calibrating the covariance-radius family of Sec.~\ref{sec:robust} and uncovering any with high dual prices. It is a hybrid of a genetic algorithm and an artificial bee colony search. Originally for molecular potential-energy surfaces where an objective evaluation is costly, here the candidates are discrete subsets of settings. The search keeps a population of them, makes local improvements and uses scout moves to restart those that stall, with recombination for the occasional larger change. Before a conic solve is done to evaluate the collection, every move or crossover must restore target coverage. Fitness is either the optimized $\Phi$ of the chosen set or the exact dual price of a candidate. There is a Pareto-archive mode for searching size and cost in tandem and a continuous mode for radius calibration. Coefficient sharing and shot allocation are always a conic program. In the LiH column generation we take a fixed 20-setting FC pool and test each unused member, reoptimizing after adding the two with the most violations until the pool is gone. The upper values on file are $0.7823$, $0.7320$, $0.6857$, $0.6490$, $0.6250$, $0.6201$, $0.6126$, $0.6047$, $0.5935$, $0.5887$ and $0.5879$. The run also holds the feasibility-repaired brackets. The largest pool prices and round lower bounds can be found in Sec.~\ref{sec:colgendemo}.

Score enrichment is applied to any observable readable at two or more of the current settings. Nonidentity products of compatible Hamiltonian terms are put through the interval test and ranked by their separation. We then add the twelve chosen scores at every compatible setting and proceed to reoptimize. In the case of the more extensive H$_4$ work, the calculation is applied to all $65351$ nonidentity Pauli strings that are not part of the target. This involves a total of 22 conic solves as up to 24 scores are added to the 565 settings on hand with every round. Any setting or score combinations left out of the model are detailed in Supplementary Note S5.

The data on file will show the operator, witness, covariance and score basis for any given result. We can test the model by way of reconstruction from signed inputs and an independent look at the covariance, while checksums will flag any changes to the files. What we have by way of numerical evidence are the primal--dual records and measurement constants. As for end-to-end finite-shot RMSE on the H$_4$ schedules or verification via interval arithmetic, those have not been done.
\subsection{Independent batches and random setting selection}
\label{sec:stratification}
One can remove the variance that comes with the random selection of a setting by taking measurements in distinct batches. The following comparison keeps the target fragments and their limiting measurement frequencies fixed. There may be other merits to a randomized approach when one is estimating several observables.

Consider $H=\sum_{\ell=1}^M F_\ell$ where each fragment is readable at its setting and let $p_\ell$ be the probability of selecting $\ell$. For projective eigenvalue scores the variance of the value $F_\ell/p_\ell$ reported by a randomized shot is
\begin{equation}
V_{\rm rand}=\sum_\ell\frac{\Tr(\rho F_\ell^2)}{p_\ell}-\langle H\rangle_\rho^2.
\label{eq:randvar}
\end{equation}
With separate batches at the same limiting proportions the variance is constant
\begin{equation}
V_{\rm strat}=\sum_\ell\frac{\Var_\rho(F_\ell)}{p_\ell}.
\label{eq:stratvar}
\end{equation}
A Cauchy--Schwarz argument shows that the difference between the two is nonnegative:
\begin{equation}
V_{\rm rand}-V_{\rm strat}
=\sum_\ell\frac{\langle F_\ell\rangle_\rho^2}{p_\ell}
-\left(\sum_\ell\langle F_\ell\rangle_\rho\right)^2\geq0.
\label{eq:strat_gain}
\end{equation}
This holds for general scores if $\Tr(\rho F_\ell^2)$ is replaced by the corresponding second moments. In effect, this nonnegative amount is the variance attributable to random setting selection, which separate batches do away with for fixed fragments and limiting frequencies. That said, it does not preclude the utility of randomized designs for multiple targets.

\subsection{Relation to existing measurement methods}\label{sec:relations}
Equation~\eqref{eq:primal} is our means of jointly optimizing coefficient sharing and allocation for the groups in question, be they in local or Clifford-rotated stabilizer bases. One finds minimum-clique-cover approaches building nonoverlapping partitions~\cite{Verteletskyi2020MCC}, or deterministic schemes that are covariance-aware and make use of transformation freedom~\cite{Yen2023Deterministic}. Some allow for overlapped grouping so a term has several readable copies~\cite{Wu2023OverlappedGrouping}. Ref.~\cite{Yen2023Deterministic} has an iterative splitting of coefficients that covers the essentials of the overlapping-variable problem: exact reconstruction, alternating allocation, coefficient copies and variances that depend on covariance. More recent work has refined the initialization~\cite{VarSI2026} and brought in deterministic repacking for a monotone reduction in variance~\cite{Repacking2026}. We offer something complementary. The conic program will solve the joint convex problem for a finite collection and its scores, and the dual tells us the size of the gap. The settings from the above methods can be put forward as candidates in Sec.~\ref{sec:colgen}.

Then there are classical shadows~\cite{Huang2020ClassicalShadows} with their inverse measurement channel for predicting observables, or locally biased shadows that adapt to a proxy state and Hamiltonian~\cite{Hadfield2022LocallyBiased}. Derandomization selects single-qubit settings deterministically~\cite{Huang2021Derandomized}. Dual-frame methods~\cite{Fischer2024DualFrames,Malmi2024Overcomplete,Korhonen2025LocallyOptimal} are used to optimize the reconstruction from data that is overcomplete or informationally complete. But in our case the target is known and the admissible score spaces are finite. Execution is stratified and informational completeness is superfluous so long as the collection spans the target. Eq.~\eqref{eq:strat_gain} is what isolates the setting-selection variance of randomized execution at fixed fragments and frequencies.

ShadowGrouping puts forward error guarantees for its schedule~\cite{Gresch2025ShadowGrouping}, but such tail bounds are not the same as a lower bound on the optimum across all schedules from the same collection. ROGS bring together optimization of circuit counts and resources with overlapping groups~\cite{Li2025ROGS}. A certificate is no substitute for these kinds of construction. They can define a collection for conical optimization of the remaining coefficient sharing, with the dual to measure how far off the optimum is, or their rules can produce candidates for pricing. For adaptive estimation of quantum observables (AEQuO)~\cite{Shlosberg2023AEQuO}, the certificate is an evaluation of the model as it stands from the realized settings and scores rather than an optimization of every policy.

Fermionic approaches use electronic structure that a fixed Pauli expansion may obscure. Fluid fermionic fragments~\cite{Choi2023Fluid} can reduce variance by redistributing terms with the help of algebraic flexibility and proxy-state information, while matchgate and fermionic Gaussian shadows give a fermion-native randomized ensemble~\cite{Wan2023Matchgate,Heyraud2024MatchgateUnified}. Their scores fit into our readable spaces. So long as the embeddings, covariance prescriptions and costs are consistent, Pauli, fermionic and hardware-native settings can be part of a single optimization.

\subsection{Relationship to experimental design and reconstruction}
In a sense, the optimization is a multiresponse $c$-optimal design problem where the shots are repetitions, the score vectors are correlated responses and the settings are the experiments~\cite{Sagnol2011,DetteHollandLetz2009}. Geometric optimality conditions and witnesses are provided by the classical equivalence theorems~\cite{KieferWolfowitz1960,Pukelsheim2006}, and there is related literature on linear unbiased estimation and optimal quantum processing~\cite{DArianoPerinotti2007,Zhu2014BLUE}. Classical design has also established how to search for design points that have been left out. Our column generation in Sec.~\ref{sec:colgen} is a quantum version of that for Pauli or stabilizer candidates. The model only calls for what is necessary to put the selected target back together. It does not have to be informationally complete in the sense of being able to reconstruct the whole quantum state or any observable.

\section*{Author contributions}
F.Z. conceived the study, carried out the derivations and calculations, and developed the optimization and verification scripts. V.-A.G. contributed to the problem formulation and the interpretation of the results. F.Z. drafted the manuscript. Both authors revised the text and approved the final version.

\section*{Competing interests}
The authors declare no competing interests.
\section*{Acknowledgments}
The authors acknowledge partial support by ORNL's LDRD and VSO programs.  This manuscript has been authored by UT-Battelle, LLC, under contract DE-AC05-00OR22725 with the U.S. Department of Energy (DOE).  This work was also supported by the U.S. Department of Energy, Office of Science, through the Ames National Laboratory under Contract No.\ DE-AC02-07CH11358.  The US government retains and the publisher, by accepting the article for publication, acknowledges that the US government retains a nonexclusive, paid-up, irrevocable, worldwide license to publish or reproduce the published form of this manuscript, or allow others to do so, for US government purposes.  DOE will provide public access to these results of federally sponsored research in accordance with the DOE Public Access Plan (https://www.energy.gov/doe-public-access-plan).  This research used resources of the Oak Ridge Leadership Computing Facility at the Oak Ridge National Laboratory, supported by the Office of Science of the U.S. Department of Energy under Contract No.\ DE-AC05-00OR22725, and resources of the National Energy Research Scientific Computing Center (NERSC), a Department of Energy Office of Science User Facility, using NERSC award m4621 (ERCAP0036406).

\bibliographystyle{unsrtnat}
\begingroup
\setlength{\bibsep}{3pt}
\bibliography{references}
\endgroup
\clearpage
\includepdf[pages=-]{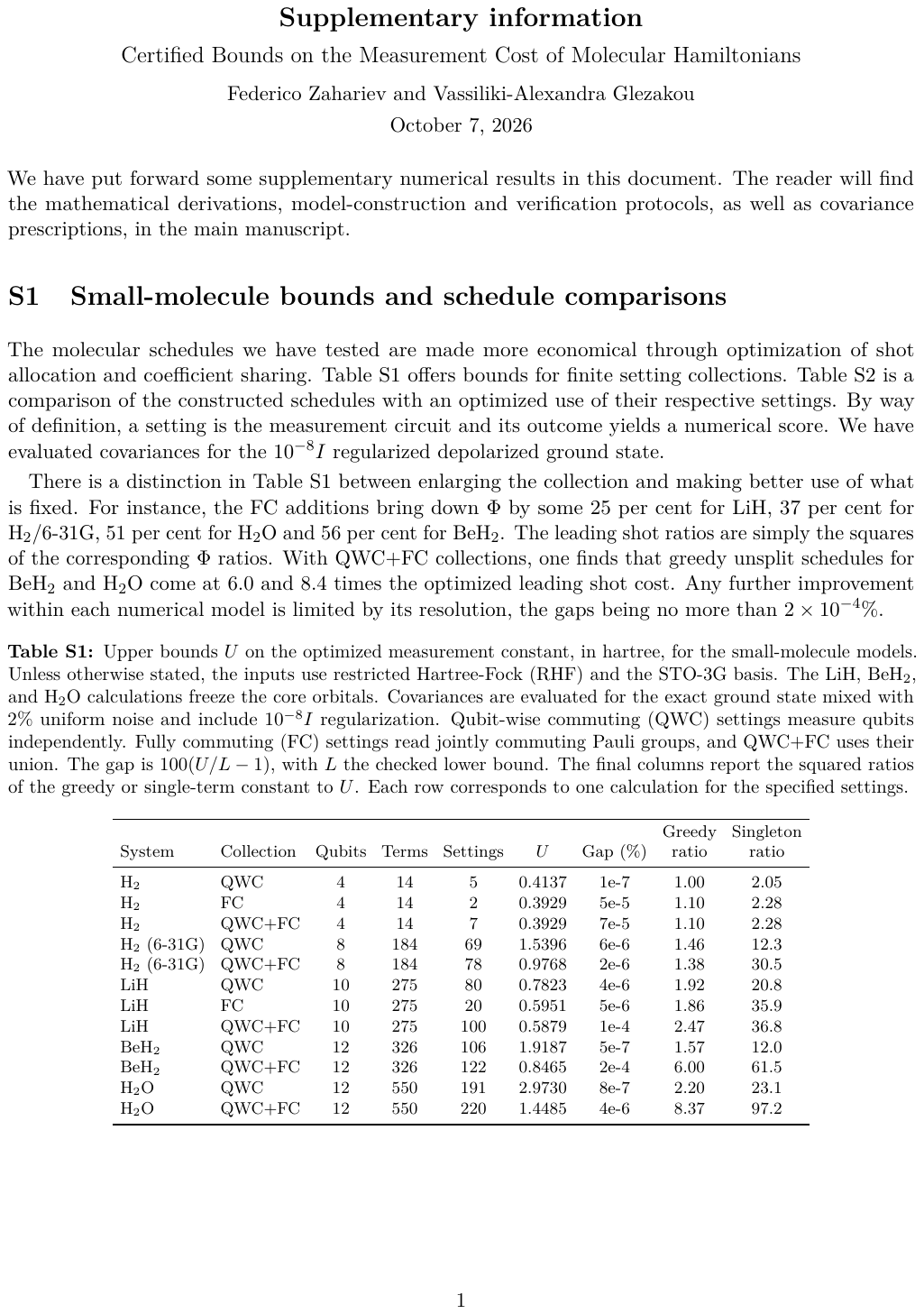}
\end{document}